\documentclass[11pt,oneside]{article}

\usepackage[round,colon,authoryear]{natbib} 

\usepackage{graphicx}
\usepackage{amsfonts}
\usepackage{amssymb}
\usepackage{amsmath}
\usepackage{amsthm}
\usepackage[shortlabels]{enumitem}
\usepackage{times}
\usepackage{eurosym}
\usepackage{url}
\usepackage{color}
\usepackage{verbatim}

\RequirePackage[colorlinks,citecolor=blue,urlcolor=blue]{hyperref}

\newcommand{\R}{\mathbb{R}}
\newcommand{\E}{\mathbb{E}}
\newcommand{\N}{\mathbb{N}}
\newcommand{\Prob}{\mathbb{P}}
\newcommand{\Q}{\mathbb{Q}}

\newcommand{\1}{\mathbf{1}}

\newcommand{\dd}{\mathrm{d}}

\newcommand{\lc}{[\![}
\newcommand{\rc}{]\!]}

\newtheorem{thm}{Theorem}[section]
\newtheorem{prop}[thm]{Proposition}

\theoremstyle{remark}
\newtheorem{ex}[thm]{Example}
\theoremstyle{remark}
\newtheorem{defin}[thm]{Definition}
\theoremstyle{remark}

\theoremstyle{remark}
\newtheorem{remark}[thm]{Remark}
\theoremstyle{remark}

\allowdisplaybreaks

\title{Financial Models with Defaultable Num\'eraires\thanks{
We thank Monique Jeanblanc, Martin Larsson, Martin Schweizer, and Hideyuki Takada for helpful discussions. We are grateful to the Referees and Associate Editor for their very careful reading and their insightful  comments. The research of S.P. benefited from the support of the Chair Markets in Transition (F\'ed\'eration Bancaire Fran\c caise) and the project ANR 11-LABX-0019. The research of S.P. has also received funding from the European Research Council under the European Union's Seventh Framework Programme (FP/2007-2013) / ERC Grant Agreement n.~307465-POLYTE.
 J.R.~acknowledges generous support from the Oxford-Man Institute of Quantitative Finance, University of Oxford.  The views represented herein are the authors' own views and do not necessarily represent the views of Morgan Stanley or its affiliates and are not a product of Morgan Stanley research.
}}
\author{
    Travis Fisher\thanks{No affiliation, E-mail: traviswfisher@gmail.com.} \and
    Sergio Pulido\thanks{Laboratoire de Math\'ematiques et Mod\'elisation d'\'Evry (LaMME), Universit\'e d'\'Evry-Val-d'Essonne, ENSIIE, Universit\'e Paris-Saclay, UMR CNRS 8071, IBGBI 23 Boulevard de France, 91037 \'Evry Cedex, France, Email: sergio.pulidonino@ensiie.fr.} \and
    Johannes Ruf\thanks{
Department of Mathematics, London School of Economics and Political Science, Columbia House, Houghton St, London WC2A 2AE, United Kingdom, E-mail:  j.ruf@lse.ac.uk.           }
}

\begin{document}
\thispagestyle{plain} \maketitle
\begin{abstract}
\noindent 
Financial models are studied where each asset may potentially lose value relative to any other. Conditioning on non-devaluation, each asset can serve as proper num\'eraire and classical valuation rules can be formulated.  It is shown when and how these local valuation rules can be aggregated to obtain global arbitrage-free valuation formulas.
\begin{description}
\item[Keywords:] Defaultable num\'eraire; Devaluation;  Non-classical valuation formulas.
\item[AMS Subject Classification (2010):] 60G48, 60H99, 91B24, 91B25, 91B70, 91G20, 91G40.
\item[JEL Classification:] D53, G13.
\end{description}
\end{abstract}

\section{Introduction}

Classical models of financial markets are built on a family of stochastic processes describing the random dynamics throughout time of the underlying assets' prices in units of a pre-specified num\'eraire.   Such a num\'eraire, often also interpreted as a money market account, is an asset that cannot devalue. In this paper we cover the case where there are multiple financial assets, any of which may potentially lose all value relative to the others. Thus, none of these assets  can serve as a proper num\'eraire. 

Pricing models for  contingent claims that allow for the devaluation of the underlying assets are numerous. For example, they appear naturally in credit risk.   In the terminology introduced by \citet{Schoenbucher_2003, Schoenbucher_2004} such assets are called {\it defaultable num\'eraires}.\footnote{The term ``defaultable num\'eraire'' sometimes appears in the credit risk literature with a different meaning, namely to describe assets with strictly positive but not adapted price processes; for example, \citet{Bielecki:Jeanblanc:Rutkowski:2004} use this definition; see also \cite{Brigo:2005}. Here, however, we follow the spirit of  \citet{Schoenbucher_2003, Schoenbucher_2004}, where a defaultable num\'eraire is a nonnegative semimartingale that has positive probability to become zero. In contrast, a proper num\'eraire shall be a strictly positive semimartingale.}  \citet{Jarrow:Yu}, \citet{Collin:2004}, and \citet{Jamshidian:2004} are further examples of this literature. Financial models for foreign exchange provide examples of other types of assets that might devalue due to the possibility of {hyperinflation} occurring; see, for example, \citet{Camara:Heston}, \citet{CFR2011}, and \cite{Kardaras:2015:valuation}.   

Another class  of models that has drawn much attention  involves strict local martingale dynamics for the asset price processes; see, for example, \citet{Sin} and \citet{HLW}.  Often such models are chosen because they can be interpreted as {bubbles} (\cite{Protter_2013_bubbles}) or they are analytically tractable (\cite{PH_hedging}, \cite{CFR_qnv}).   Both practitioners (\cite{Lewis}, \cite{Paulot}) and academics (\cite{CH}, \cite{MadanYor_Ito}) suggest {non-classical valuation formulas} for contingent claims in such models  in order to be consistent with market prices.  In this paper, we argue that strict local martingale dynamics are consistent with the interpretation that the corresponding num\'eraire devalues. This point of view then enables us to interpret the correction term in  the valuation formula of  \citet{Lewis}    as the value  of the contingent claim's payoff in the scenarios where the num\'eraire devalues. Thus, the valuation formulas of \citet{Lewis}, \cite{MadanYor_Ito},  \cite{Paulot}, or  \citet{Kardaras:2015:valuation} arise as special cases of this paper's framework.  Parallel to this paper, \cite{Herdegen:Schweizer:2017} have developed an alternative consistent and general valuation framework that always guarantees put-call parity.
 
 This paper's contributions can now be summarized as follows.
 \begin{enumerate}
 	\item It provides an interpretation of strict local martingale models, which can arise by fixing a num\'eraire that has positive probability to default.  Non-classical valuation formulas, restoring put-call parity, can then be economically justified and extended.
	\item Assume, for the moment, that for each asset there exists a  probability measure under which  the discounted prices (with the corresponding asset as num\'eraire) are local martingales (or, even,  supermartingales).  These measures  need not be equivalent. This paper provides conditions under which these measures can be \emph{aggregated} to an arbitrage-free valuation operator that takes all events of devaluations into account.
 \end{enumerate}
 
 In Section~\ref{S:2}, we introduce the framework.   We consider a model for $d$ assets.  We have in mind a foreign exchange market, and hence call these assets ``currencies", but really these could represent any asset with non-negative value.  We denote the value of one unit of the $j$-th currency, measured in terms of the $i$-th currency, as $S_{i,j}$.  We model the full matrix $(S_{i,j})_{i,j}$ of these exchange rates.  This is convenient because our main results are formulated in terms of the relative prices contained in the  {\it matrix of exchange rates}. If the $j$-th currency has devalued with respect to the $i$-th currency  at time $t$ we have $S_{i,j}(t) = 0$ and $S_{j,i}(t) = \infty$. In this case, the $j$-th currency cannot be used as a num\'eraire, and the standard results of mathematical finance in units of this currency do not apply. In this setup, however, the basket can be used as a num\'eraire. That is, the portfolio consisting of one unit of each currency with value $\sum_{i,j} S_{i,j}$, measured in terms of the $i$-th currency, is a valid num\'eraire.

Section~\ref{sec:aggreg} contains the paper's main contributions. First, families of \emph{num\'eraire-consistent} probability measures are considered.  Each of these measures corresponds, in a certain sense, to 
fixing a specific currency as the underlying num\'eraire. We call \emph{disaggregation} the step that constructs this family of  \emph{num\'eraire-consistent} probability measures from a valuation  measure when the basket is taken as num\'eraire. We call \emph{aggregation} the reverse step, namely taking a possibly non-equivalent family of probability measures, corresponding to the different currencies as num\'eraires, and constructing a valuation  measure for the basket.  
Embedding a strict local martingale model in a family of num\'eraire-consistent probability measures and then aggregating this family  yields the non-classical valuation formulas of \citet{Lewis}, \citet{MadanYor_Ito}, \citet{Paulot}, and \citet{CFR2011}.  This point of view has two advantages. First of all, it yields generic valuation formulas for any kind of contingent claim. These formulas are consistent with the above-mentioned non-classical valuation formulas, which are usually only provided for specific claims. Second, it gives an economic interpretation to the lack of martingale property as the possibility of a default of the underlying num\'eraire.

Section~\ref{S:proofs} contains the proofs of the main results.   

We point out the recent work of \cite{Tehranchi:2015}, who considers an economy where prices quoted in terms of a given non-traded currency are not necessarily positive. Relative prices between the assets are not studied. Instead,  \cite{Tehranchi:2015} focuses on different arbitrage concepts taking into consideration that the agent might not be able to substitute today's consumption by tomorrow's consumption.

\subsection*{Notation}
Throughout the paper we fix a deterministic time horizon $T>0$ and consider an economy with $d \in \N$ traded assets, called ``currencies''.  To reduce notation, we shall use the generic letter $t$ for time and abstain from using the qualifier ``$\in [0,T]$''.  We shall also use the generic letters $i,j,k$ for the currencies and again abstain from using the qualifier ``$\in \{1, \cdots, d\}$". For example, we shall write ``$\sum_{j}$'' to denote ``$\sum_{j=1}^d$''.  When introducing a process $X = (X(t))_{t \in [0,T]}$, we usually omit ``$= (X(t))_{t \in [0,T]}$''. If $v \in \R^d$, we understand inequalities of the form $v \geq 0$ componentwise. For a matrix $\Gamma\in\R^{d\times d}$, we shall denote by $\Gamma_i$ the $i$-th row of $\Gamma$. Moreover, we use the convention $\inf \emptyset = \infty$ and we denote the cardinality of a countable set $A$ by $|A|$. Furthermore, we emphasize that a product $xy$ of two numbers $x,y \in [0,\infty]$ is always defined except if either (a) $x=0$ and $y = \infty$ or (b) $x=\infty$ and $y = 0$.

We fix a filtered space $(\Omega,\mathcal{F}(T),(\mathcal{F}({t}))_{t})$, where the filtration $(\mathcal{F}({t}))_{t}$ is assumed to be right-continuous and $\mathcal F(0)$ to be trivial. In the absence of a probability measure, all statements involving random variables or events are supposed to hold pathwise for all $\omega\in\Omega$. For an event $A \in \mathcal F(T)$, we set $\1_A(\omega) \times \infty$ and $\1_A(\omega) \times(- \infty)$ to $\infty$ and $-\infty$, respectively, for all $\omega \in A$ and to $0$ for all $\omega \notin A$. Let us now consider a probability measure $\Q$ on $(\Omega, \mathcal F(T))$. We write $\E^{\Q}$ for the corresponding expectation operator and $\E^{\Q}_t$ for the conditional expectation operator, given $\mathcal F(t)$, for each $t$. We let $L^{1}(\Q)$ denote the space of (equivalence classes of) real-valued random variables $Z$ such that $\E^{\Q}[|Z|]<\infty$. For a real-valued semimartingale $X$ with $X(0) = 0$ we write $\mathcal E(X)$ to denote its stochastic exponential; that is,
\begin{align*}
	\mathcal E(X) = \mathrm e^{X-[X,X]^c/2} \prod_{s \leq \cdot} (1+\Delta X_s) \mathrm e^{- \Delta X_s},
\end{align*}
where $\Delta X = X-X_-$ and $[X,X]^c$ denotes the quadratic variation of the continuous part of  $X$.

\section{Framework}  \label{S:2}
This section introduces the concept of \emph{exchange matrices} to represent prices of the underlying currencies and the related concept of \emph{value vectors} to describe valuations of contingent claims with the currencies as underlying. 

We put ourselves in an economy that is characterized by the price processes of $d$ currencies  relative to each other  via an $[0,\infty]^{d \times d}$--valued, right-continuous, $(\mathcal F(t))_t$--adapted process $S =  (S_{i,j})_{i,j}$. Here, the process $S_{i,j}$ denotes the price process of the $j$-th currency in units of the $i$-th currency. We also refer to \citet{Vecer_2011}, where a similar point of view is taken.  In order to simplify the analysis below we assume that interest rates are zero. Alternatively, we might interpret $S_{i,j}(t)$ as the price of one unit of the $j$-th money market in terms of units of the $i$-th money market at time $t$, for each $i,j,$ and $t$.

In order to provide an economic meaning to the matrix-valued process $S$ we shall assume that it satisfies certain consistency conditions. Formally, we assume that $S(t)$  is an
 exchange matrix for each $t$, in the sense of the following definition.
\begin{defin}[Exchange matrix]\label{D:1}
	An exchange matrix is a $d\times d$-dimensional
matrix $s=(s_{i,j})_{i, j}$ taking values in $[0,\infty]^{d \times d}$ with the property
that $s_{i,i} = 1$ and $s_{i,j}s_{j,k}=s_{i,k}$ for all $i,j,k$, whenever the product is defined.
\qed
\end{defin}
Note that the definition implies, in particular, that an exchange matrix $s$ also satisfies that $s_{i,j} = 0$ if and only if $s_{j,i} = \infty$ for all $i,j$.  The consistency conditions of Definition~\ref{D:1} guarantee the following: for fixed $i,j,k$, an investor who wants to exchange units of the $i$-th currency into units  of the $k$-th currency is indifferent between exchanging directly $s_{i,k}$ units of the $i$-th currency into the $k$-th currency or, instead, going the indirect way and first exchanging the appropriate amount of units of the $i$-th currency into the $j$-th currency and then exchanging those units into the $k$-th currency.

For each $t$, we define the index set of ``active currencies''
	$$\mathfrak{A}(t) = \left\{i: \sum_j S_{i,j}(t) < \infty\right\}.$$
If $i \in \mathfrak{A}(t)$ for some $t$ then the $i$-th currency has not devalued against any other currency. To simplify notation, we shall assume that $\mathfrak A(0) = \{1, \cdots, d\}$; that is, at time $0$ no currency has devalued.

\begin{remark}[Existence of a strong currency]\label{r:finiterow} We always have $\mathfrak{A}(t) \not= \emptyset$ for each $t$. More precisely, if $s$ is an exchange matrix, there exists $i$ such that $s_{i,j}\leq 1$ for all $j$. To see this, we define, on the set of indices $\{1,\ldots,d\}$, a total preorder as follows: $j\preceq k$ if and only if $s_{j,k}\geq 1$, that is, if and only if the $k$-th currency is ``stronger'' than the $j$-th currency. The consistency conditions of Definiton~\ref{D:1} guarantee that  this is a total preorder. As the set of indices is finite, there exists a (not necessarily unique) maximal index $i$ corresponding to the ``strongest'' currency. For such an index $i$ we have   $s_{i,j}\leq 1$ for all $j$. \qed
\end{remark}

We are interested in additional assets in the economy besides the $d$ currencies and in their relative valuation with respect to those currencies. Towards this end, we introduce the notion of value vector.

\begin{defin}[Value vector for exchange matrix]  \label{D:2}
	A value vector for an exchange matrix $s$  is a $d$-dimensional vector $v=(v_{i})_{i}$ taking values in $[0,\infty]^{d}$ with the property that $s_{i,j}v_{j}=v_{i}$ for all $i,j$,
whenever the product is defined.
\qed
\end{defin}

	A value vector encodes the valuation of an asset in terms of the $d$ currencies.  More precisely, the $i$-th component describes how many units of the $i$-th currency are required to obtain one unit of that specific asset.   The consistency condition in Definition~\ref{D:2} guarantees again that an investor who wants a unit of the new asset does not prefer to first exchange her currency into another one in order to obtain that asset. We shall write
\begin{align*}
	\mathcal{C} &=\Bigg\{C: \quad \text{$C$ is an $\mathcal F(T)$--measurable value vector for $S(t)$}  \Bigg\}.
\end{align*}

For all $i$ we denote by $I^{(i)}$ the value vector corresponding to the value of one unit of the $i$-th currency at time $T$ in terms of the other currencies, that is
    \begin{equation}\label{eq:vector_I}
        I^{(i)}=(S_{j,i}(T))_j.
    \end{equation}

We now consider the financial market where prices are quoted in terms of the basket asset -- the portfolio consisting of one unit of each currency. 
The value of the basket in terms of the $i$-th asset equals exactly $\sum_{j} S_{i,j}$ and the value of the $i$-th asset in terms of the basket equals its reciprocal. 
This simple observation allows us to express the relative price process $\overline{S}=(\overline{S}_i)_i$ with respect to the basket num\'eraire  as
\begin{equation}\label{def:basketmkt}
\overline{S}_i=\frac{1}{\sum_{j}S_{i,j}}.
\end{equation}
Observe that in~\eqref{def:basketmkt} there are no divisions by zero because $\sum_{j}S_{i,j}\geq  S_{i,i} = 1$; hence, $0\leq \overline{S}_i\leq 1$ for all $i$.
Vice versa, note that we  get
\begin{align}  \label{eq:160727.1}
	S_{i,j} = \frac{S_{i,j} \left(\sum_{k}S_{i,k}\right)^{-1}}{\overline S_i}  = \frac{ \left(\sum_{k}S_{j,k}\right)^{-1}}{\overline S_i}  =  \frac{\overline S_j}{\overline S_i},
\end{align}
whenever the fractions are well-defined, for all $i,j$.
Thanks to~\eqref{eq:160727.1}, with $i^* \in \mathfrak A(t)$, we get
\begin{align} \label{eq:170310.3}
	\sum_{j} \overline S_j(t)= \sum_{j\in \mathfrak A(t)} \frac{S_{i^*,j}(t)\overline S_j(t)}{S_{i^*,j}(t)} =  \overline S_{i^*}(t)\sum_{j\in  \mathfrak A(t)}  S_{i^*,j}(t) = 1
\end{align}
for all $t$; hence, in the basket market an asset representing a risk-free bond can be replicated.

Next, given a value vector $C \in \mathcal C$, we let $\overline C$ be the payoff of the claim $C$ in terms of the basket. More precisely, $\overline C = \overline S_i(T) C_i$ for all $i\in\mathfrak A(T)$. This quantity is well-defined thanks to~\eqref{eq:160727.1} and we have the representation 
\begin{equation}\label{def:basketvalue}
\overline{C}=\frac{1}{{|\mathfrak A(T)|}}\sum_{j \in\mathfrak A(T) }\overline{S}_j(T) C_j.
\end{equation}
As $0<\overline{S}_i(T)\leq 1$ on $\{i\in\mathfrak A(T)\}$, all the multiplications in \eqref{def:basketvalue} are well-defined. To understand the representation in \eqref{def:basketvalue} better, first note that inserting the value vector $C = I^{(i)}$ of \eqref{eq:vector_I} leads exactly to $\overline C = \overline S_i(T)$, for each $i$.\footnote{Indeed, with $C = I^{(i)}$,  \eqref{def:basketvalue} becomes, with the help of \eqref{eq:160727.1},  for each $i$,
 \begin{align*}
 	\overline{C}=\frac{1}{{|\mathfrak A(T)|}}\sum_{j\in\mathfrak A(T)} \overline S_j (T) S_{j,i}(T)  = \frac{1}{{|\mathfrak A(T)|}}\sum_{j\in\mathfrak A(T)} \overline S_i(T)
	= \overline S_i (T).
 \end{align*}
 }
 More generally, the formula in \eqref{def:basketvalue} transforms, for each active currency $i \in \mathfrak A(T)$ the payoff $C_i$ into a payoff measured with respect to the basket, by dividing with the value of the basket. All of the  $|\mathfrak A(T)|$ values are identical, thus summing them up and dividing by  $|\mathfrak A(T)|$ does not modify the value, which can be interpreted as the payoff of the corresponding contingent claim, measured in terms of the basket. The decomposition of $\overline C$ given in~\eqref{def:basketvalue} is crucial to obtain the aggregation results of Section~\ref{SS:4.1}.
 
We have now translated the setup of this section into a classical setup, with $\overline S$ denoting a vector-valued price process, measured in terms of a basket as in \cite{Yan_new}, and $\overline C$ denoting a one-dimensonal random variable, representing the payoff of a contingent claim, again measured in terms of the basket.  

\begin{remark}
We could have started our study from a more classical setup where the prices of the $d$ assets in the market, quoted in terms of an external num\'eraire, are given by a vector-valued process $X=(X_1,\ldots,X_d)$. In this context, the relative prices would correspond to $S_{i,j}=X_j/X_i$ whenever the fractions are well-defined. If the external num\'eraire is given by the basket, as shown in \eqref{eq:160727.1}, this is precisely the framework described above where the basket-quoted prices are denoted by $\overline{S}$. \qed
\end{remark}

We say that a probability measure $\Q_i$ is a valuation measure with respect to $i$-th currency if the prices $S_i$ are supermartingales with respect to $\Q_i$.   A probability measure $\overline \Q$ on $(\Omega,\mathcal F(T))$ is a valuation measure with respect to the basket if the basket-quoted prices $\overline{S}$ are local martingales with respect to $\overline\Q$. As $0\leq \overline{S}_i\leq 1$ for all $i$, $\overline\Q$ is a valuation measure with respect to the basket if and only if $\overline{S}$ is a $\overline\Q$--martingale.

\section{Aggregation and disaggregation of measures}
\label{sec:aggreg}
We now investigate how to aggregate a family $(\Q_i)_i$ of valuation measures, each supported on a subset of the set $\Omega$ of possible scenarios and relative to one of the $d$ currencies, to a valuation measure $\overline \Q$ with respect to the basket.  We provide the proofs of this section's assertions in Section~\ref{S:proofs}.  We structure this study in four parts. 

Subsection~\ref{SS:4.1} discusses how the existence of a valuation measure $\overline \Q$ with respect to the basket yields a family of $d$ probability measures, which are not necessarily equivalent. Each of these $d$ measures can be interpreted as a valuation measure with one of the $d$ num\'eraires fixed.   Moreover, the measures are related to each other via a generalized change-of-num\'eraire formula. This property is called \emph{num\'eraire-consistency}.   We then show that if a family of probability measures is num\'eraire-consistent they can be ``stuck together'' to yield a valuation measure $\overline \Q$ with respect to the basket.

Subsection~\ref{SS:4.1.5} compares the results of Subsection~\ref{SS:4.1} with the classical valuation formulas in mathematical finance.
 Subsection~\ref{SS:4.2}  provides two examples.  They illustrate, in particular, how the results of \citet{CFR2011} and \citet{Camara:Heston} are special cases of this paper's setup. 
In Subsection~\ref{SS:4.3} we start with $d$ probability measures, each serving again as a valuation measure for a fixed num\'eraire. However, this time we do not assume that these measures are num\'eraire-consistent. We then study conditions such that a valuation measure $\overline \Q$ with respect to the basket exists.

\subsection{Aggregation with num\'eraire-consistency and disaggregation}  \label{SS:4.1}

We start by introducing and discussing the following consistency condition.
\begin{defin}[Num\'eraire-consistency of probability measures]  \label{D:consistency}
Suppose that $(\Q_i)_i$ is a family of probability measures. We say that $(\Q_i)_i$ is a num\'eraire-consistent family of probability measures if for all  $A \in \mathcal F(t)$ we have
    \begin{equation}\label{eq:consistent2a}
        \E^{\Q_i}[S_{i,j}(t) \1_A]= S_{i,j}(0) {\Q_j}(A \cap {\{S_{j,i}(t)>0\}})
    \end{equation}
 for all $i,j$ and $ t$.
\qed
\end{defin}

\begin{prop}[Properties of a num\'eraire-consistent family of probability measures]  \label{P:cons}
Suppose that $(\Q_i)_i$ is a num\'eraire-consistent family of probability measures. Then the following statements hold, for each $i,j$.
\begin{enumerate}[label={\rm(\alph{*})}, ref={\rm(\alph{*})}]
\item  \label{P:cons:1} $S_i$ is a $\Q_i$--supermartingale; thus, in particular, $\Q_i(\bigcap_t\{i\in{\mathfrak{A}(t)}\})=1$. More precisely, we have
    \begin{equation}\label{eq:consistent2b}
        \E_r^{\Q_i}[S_{i,j}(t)X]= S_{i,j}(r) \E_r^{\Q_j}[X \1_{\{S_{j,i}(t)>0\}}], \qquad  \text{$\Q_i$--almost surely} 
    \end{equation}
     for all bounded $\mathcal F(t)$--measurable random variables $X$ and $r \leq t$.
\item \label{P:cons:2}
$S_{i,j}$ is a $\Q_i$--local martingale  if and only if $S_{j,i}$ does not jump to zero under $\Q_j$.
\item \label{P:cons:3}
For each stopping time $\tau$, 
 $S_{i,j}^\tau$ is a $\Q_i$--martingale  if and only if $\Q_j(S_{j,i}(\tau) > 0) = 1$. Moreover, in this case we have $\dd \Q_j /\dd \Q_i|_{\mathcal{F}(\tau)} = S_{j,i}(0) S_{i,j}(\tau) $.   In particular, the $i$-th currency does not completely devalue with respect to the $j$-th currency, if and only if $S_{i,j}$ is a true $\Q_i$--martingale.
\end{enumerate}
\end{prop}

Note that \eqref{eq:consistent2b} can be interpreted as a change-of-num\'eraire formula.

\begin{remark}[An interpretation for num\'eraire-consistency]\label{R:int}
Let  $(\Q_i)_i$ be a num\'eraire-consistent family of probability measures. Then
 with $w_{i,j} = S_{i,j}(0) / \sum_k  S_{i,k}(0) \in (0,1)$ for all $i,j$,  we have $\sum_j w_{i,j} = 1$ and
     \begin{align*}
         1-  \frac{1}{  \sum_k  S_{i,k}(0)}  \E^{\Q_i}\left[\sum_j S_{i,j}(T)\right]  &=  \sum_j w_{i,j} \left(1-S_{j,i}(0) \E^{\Q_i}[S_{i,j}(T)] \right) =\sum_j w_{i,j} \Q_j(S_{j,i}(T)= 0)
    \end{align*}
    for all $i$.  Therefore,  the normalized expected decrease of the total value of all currencies, measured in terms of the $i$-th currency, equals the sum of  the weighted probabilities that the $i$-th currency completely devalues.  The weights correspond exactly to the proportional value of the corresponding currency at time zero. \qed
\end{remark}

We are now ready to formulate a first aggregation result for a num\'eraire-consistent family of probability measures. 

\begin{thm}[Aggregation and disaggregation]\label{T:extending_val_operator_prelim}
The following statements hold.
\begin{enumerate}[label={\rm(\alph{*})}, ref={\rm(\alph{*})}]
	\item\label{T:extending_val_operator_prelim:a} Given a  valuation measure  $\overline \Q$ with respect to the basket there exists a  unique num\'eraire-consistent family of probability measures $(\Q_i)_i$ such that $\sum_i\Q_i  \sim \overline \Q$ and
	\begin{equation}\label{eq:extension_val_operator}
       \E^{\overline \Q}_r [\overline C]=\sum_{j \in \mathfrak{A}(r)} \overline S_{j}(r)\E^{\Q_{j}}_r\left[\frac{C_j}{|\mathfrak{A}(t)|}\right]
    \end{equation}
for all $r$ and $C\in \mathcal{C}$ such that $\overline C\in L^{1}(\overline \Q)$.  
\item\label{T:extending_val_operator_prelim:b} Given a num\'eraire-consistent family of probability measures $(\Q_i)_i$  there exists a unique valuation measure  $\overline \Q \sim \sum_i\Q_i$   with respect to the basket that satisfies \eqref{eq:extension_val_operator} for all $r$  and $C\in \mathcal{C}$ such that $C_i\in L^1(\Q_i)$ for all $i$.
\item\label{T:extending_val_operator_prelim:c}  Consider a  valuation measure  $\overline \Q$ with respect to the basket. Let  $(\Q_i)_i$ be the corresponding num\'eraire-consistent family of probability measures from \ref{T:extending_val_operator_prelim:a} and fix $r$. If a contingent claim $C \in \mathcal C$ satisfies   $\overline C = \overline C \1_{\{i \in \mathfrak A(T)\}}$, $\overline \Q$--almost surely,  for some $i$, and $\overline{C}\in L^{1}(\overline \Q)$ then we have
\begin{align}  \label{eq:220815}
	 \E^{\overline \Q}_r [\overline C]= \overline S_{i} (r) \E^{\Q_i}_r[C_i].
\end{align}
\end{enumerate}
\end{thm}
The proof of Theorem~\ref{T:extending_val_operator_prelim} reveals the relationships  $
{\dd\Q_i}/{\dd \overline \Q}={\overline{S}_i(T)}/{\overline{S}_i(0)}$ and $\overline \Q = \sum_i \overline{S}_i(0)  {\Q_i}$.
Let us interpret the representation in \eqref{eq:extension_val_operator}. In order to compute the valuation of a contingent claim  $C\in \mathcal{C}$ under a valuation measure  $\overline \Q \sim \sum_i\Q_i$   with respect to the basket one can proceed according to the following steps.  First, one replaces the claim $C$ by the claim $\widetilde C = C/|\mathfrak A(T)|$; to wit, one divides the payoff of the contingent claim by the number of active currencies at maturity $T$.  Then, one computes the  expectation of this payoff under $\Q_j$ corresponding to fixing the $j$-th currency as num\'eraire, for each $j$.  One then converts all these values into the basket and adds them up. If the contingent claim $C$ has zero payoff in the case that the $i$-th currency completely devalues, then \eqref{eq:220815} holds. Hence, in this case one can compute $ \E^{\overline \Q} [\overline C]$  by only computing the corresponding valuation expectation with the $i$-th currency as num\'eraire.

In the terminology of  \citet{Schoenbucher_2003, Schoenbucher_2004}, $\Q_i$ is called a ``survival measure'' (corresponding to the $i$-th currency) as it is equivalent to the probability measure $\overline \Q$, conditioned on the $i$-th currency not completely devaluating.

\subsection{Comparison with the classical setup} \label{SS:4.1.5} 
In this subsection we compare the aggregation results of Theorem~\ref{T:extending_val_operator_prelim} with the classical valuation formulas in mathematical finance. As it is customary in the classical setup, we fix a num\'eraire; say, the one corresponding to the first row of the exchange matrix. Moreover, we assume that $S_1$, the vector of prices quoted with respect to this num\'eraire, is  a $\Prob_1$--semimartingale for some probability measure $\Prob_1$.  In order to clarify this comparison we now consider three different cases. 
\begin{enumerate}
\item Suppose that there exists a probability measure $\Q_1\sim\Prob_1$ such that $S_1$ is a $\Q_1$--martingale and that $S_1>0$ under $\Prob_1$. In the terminology of~\cite{Yan_new}, in this case the market is \emph{fair}. In particular, NFLVR for admissible strategies, as defined in~\cite{DS_fundamental}, holds with respect to $\Prob_1$. We can define a num\'eraire-consistent family of probability measures $(\Q_i)_i$ through the change-of-num\'eraire formula
\begin{equation}\label{eq:classicalRN}
\frac{\dd \Q_i}{\dd \Q_1} = S_{i,1}(0) S_{1,i}(T).
\end{equation}
Indeed, we have
\begin{displaymath}
\E^{\Q_i}[S_{i,j}(t) \1_A]= S_{i,1}(0)\E^{\Q_1}[S_{1,i}(t)S_{i,j}(t) \1_A]=S_{i,1}(0)\E^{\Q_1}[S_{1,j}(t) \1_A]=S_{i,j}(0){\Q_j}(A)
\end{displaymath}
for all $i,j$, $t$, and $A\in\mathcal F(t)$. Observe that $\Q_i\sim\Q_j$  for all  $i,j$ and $\sum_i \Q_i  \sim\Prob_1$.  In this case, we have \eqref{eq:220815} for all $i$, $r$, and $C \in \mathcal C$ such that $C_i\in L^{1}(\Q_i)$ for all $i$.

\item Suppose now that there exists a probability measure $\Q_1\sim\Prob_1$ such that $S_1$ is a $\Q_1$--martingale but $S_1$ is not necessarily positive under $\Prob_1$. In this case, the market is also fair in the terminology of~\cite{Yan_new}.  However, the $i$-th currency might not be a classical num\'eraire, as it does not stay strictly positive. Nevertheless, an interpretation as ``defaultable num\'eraire'' is possible and we can again define a family of probability measures thorough~\eqref{eq:classicalRN}. We observe that
now $\Q_i$ is not necessarily equivalent to $\Q_1$, but only absolutely continuous with respect to $\Q_1$ for all $i$ and $\sum_i \Q_i \sim\Q_1$. The family of probability measures is num\'eraire-consistent because
\begin{displaymath}
\begin{split}
\E^{\Q_i}[S_{i,j}(t) \1_A]&= S_{i,1}(0)\E^{\Q_1}[S_{1,i}(t)S_{i,j}(t) \1_{A\cap\{S_{i,j}(t)<\infty\}}]=S_{i,j}(0)\Q_j(A\cap\{S_{i,j}(t)<\infty\})\\&=S_{i,j}(0)\Q_j(A\cap\{S_{j,i}(t)>0\}).
\end{split}
\end{displaymath}
It is possible that $S_i$ is  not a $\Q_i$--local martingale but only a $\Q_i$--supermartingale. Moreover,  NFLVR as in~\cite{DS_fundamental} does not necessarily hold with respect to $\Q_i$.

Due to Theorem~\ref{T:extending_val_operator_prelim}\ref{T:extending_val_operator_prelim:c}, 
we have 
\begin{align} \label{eq:170310.1}
	 \E^{\overline \Q}_r [\overline C]= \overline S_1 (r) \E^{\Q_1}_r[C_1]
\end{align}
for all $r$
and $C\in \mathcal{C}$ such that $\overline{C}\in L^{1}(\overline \Q)$. With respect to the other currencies $i \geq 2$, however, \eqref{eq:220815} is not necessarily true. Indeed, \eqref{eq:170310.1} and~\eqref{eq:consistent2b} enable us to write the valuation, expressed in terms of the $i$-th num\'eraire, on the event $\{i \in \mathfrak A(r)\}$, as
\begin{displaymath}
\begin{split}
 \frac{ \E^{\overline \Q}_r [\overline C]}{\overline S_i(r)}  &= S_{i,1} (r) \E^{\Q_1}_r[C_1]=  \E^{\Q_i}_r[S_{i,1}(T)C_1]+  S_{i,1} (r) \E^{\Q_1}_r[C_1\1_{\{S_{1,i}(T)=0\}}]\\&= \E^{\Q_i}_r[ C_i]+  S_{i,1} (r) \E^{\Q_1}_r[C_1\1_{\{S_{1,i}(T)=0\}}]\end{split}
\end{displaymath}
for all $r$, $i$, and $C\in \mathcal{C}$ such that $\overline{C}\in L^{1}(\overline \Q)$. Therefore, if the valuation formula is expressed in terms of the $i$-th num\'eraire, the conditional expected value must be adjusted by a factor depending on the $\Q_1$--probability of devaluation of the $i$-th currency. Let us now fix $i$ and $C \in \mathcal C$, with $\overline{C}\in L^{1}(\overline \Q)$, and study the correction term
\begin{align*}
	F_i(r) = S_{i,1} (r)  \E^{\Q_1}_r[C_1\1_{\{S_{1,i}(T)=0\}}]
\end{align*}
for all $r$.  Then $F_i(T) = S_{i,1}(T) C_1 \1_{\{S_{i,1}(T)=\infty\}} = 0$ under $\Q_i$. 
Moreover, for $r_1 \leq r_2$ we have
\begin{align*}
	\E^{\Q_i}_{r_1}[F_i(r_2)] = S_{i,1}(r_1) \E^{\Q_1}_{r_1}\left[\1_{\{S_{1,i}(r_2)>0\}} \E^{\Q_1}_{r_2}[C_1\1_{\{S_{1,i}(T)=0\}}]\right] \leq F_i(r_1)
\end{align*}
under $\Q_i$. Hence the correction term $F_i$ is a $\Q_i$--potential.  More precisely, by the same computations, $F_i$ is indeed a  $\Q_i$--local martingale, as long as $S_{1,i}$ does not jump to zero under $\Q_1$.  Hence, under this no-jump condition, valuation  yields a local martingale, whose Riesz decomposition equals conditional expectation under $\Q_i$ plus the correction term $F_i$.

\item Suppose now that there does not exist a measure $\Q_1\sim\Prob_1$ such that $S_1$ is a $\Q_1$--martingale. In this case, according to Theorem 3.2 in~\cite{Yan_new}, NFLVR for $\Prob_1$--allowable strategies\footnote{A trading strategy is $\Prob_1$--allowable if the corresponding wealth process is bounded from below by a  constant times  $\sum_j S_{1,j}$, $\Prob_1$--almost surely.  See \cite{Yan_new} for the precise definition.} does not hold.  Let us assume, nevertheless, that there exists a family of num\'eraire-consistent probability measures $(\Q_i)_i$. As $S_1$ is only a $\Q_1$--supermartingale, and there does not necessarily exist a probability measure equivalent to $\Prob_1$ under which $S_1$ is a local martingale, NFLVR with respect to admissible strategies, as defined in~\cite{DS_fundamental}, might fail, nevertheless. However, thanks to Theorem~\ref{T:extending_val_operator_prelim}, NFLVR  holds if the basket is chosen as num\'eraire, under the probability measure $\overline \Q \sim \sum_i \Q_i$. This extended model would depart from the classical setup as  $\overline \Q(S_{1,i}(T)=\infty)>0$ for some $i$. 

The representation in \eqref{eq:extension_val_operator} provides an economically meaningful valuation formula in this setup. This is possible even if, after fixing the first currency as num\'eraire, NFLVR with respect to $\Prob_1$--allowable strategies (or even admissible strategies) is violated. The key is to consider non-equivalent probability measures that allow for devaluations of the first currency and to take consistently these states of the world into account for arbitrage considerations and valuation of options. As we illustrate in Example~\ref{ex:cosistent_2assets} below, these valuation formulas can be applied to strict local martingale models and yield values consistent with market conventions such as put-call parity. This valuation framework corrects the deficiencies of the classical conditional expectation approach for valuation. 

\end{enumerate} 

\subsection{Examples}  \label{SS:4.2}

 As already pointed out in \citet{Lewis},~\citet{CH},~\citet{MadanYor_Ito}, and~\citet{CFR2011}, among others, a strict local martingale measure is not always suitable for valuation purposes because values computed through expectations with this measure fail to be in accordance with market conventions such as put-call-parity. The works of  \citet{Lewis} and~\citet{MadanYor_Ito} propose ad-hoc correction terms to solve these deficiencies. Similarly to the study in~\cite{CFR2011}, we recognize that the problems arise from the fact that a strict local martingale measure does not take into account the states of the world where the corresponding currency devalues.  
  
\begin{ex}[The case $d=2$]\label{ex:cosistent_2assets}
Consider an economy with  $d=2$ currencies and assume a num\'eraire-consistent family $(\Q^1, \Q^2)$ exists. Indeed,  given a probability measure $\Q^1$ such that $S_{1,2}$ is a nonnegative $\Q^1$--supermartingale, a probability measure $\Q^2$ can be constructed such that $(\Q^1, \Q^2)$ is num\'eraire-consistent, for example by the approach pioneered in \citet{F1972}; see also \citet{Perkowski_Ruf_2014}.

 In the following we derive a representation of ${\overline \Q}$, the valuation measure with the basket as underlying, whose existence is guaranteed by Theorem~\ref{T:extending_val_operator_prelim}\ref{T:extending_val_operator_prelim:b}.
To this end, fix a time $r$ and  a contingent claim $C \in \mathcal C$ such that $C_i\in L^{1}(\Q_i)$ for all $i$. We then have
\begin{align}  
\E^{\overline \Q}[\overline C]&= \overline S_{1}(r)\E^{\Q_1}_r\left[\frac{C_1}{|\mathfrak{A}(T)|}\right]+\overline S_{2}(r)\E^{\Q_2}_r\left[\frac{C_2}{|\mathfrak{A}(T)|}\right] \nonumber \\
&= \overline S_{1}(r)\left(\E^{\Q_1}_r\left[\frac{C_1}{2} \1_{\{S_{1,2}(T)>0\}}\right]+\E^{\Q_1}_r\left[C_1 \1_{\{S_{1,2}(T)=0\}}\right]\right) \nonumber \\
&\quad+\overline S_{2}(r)\left(\E^{\Q_2}_r\left[\frac{C_2}{2} \1_{\{S_{1,2}(T)<\infty\}}\right]+\E^{\Q_2}_r\left[C_2 \1_{\{S_{1,2}(T)=\infty\}}\right]\right)\nonumber \\
&=\overline S_{1}(r)\E^{\Q_1}_r[C_1  ]+\overline S_{2}(r)\E^{\Q_2}_r[C_2  \1_{\{S_{1,2}(T)=\infty\}}].  \label{eq:220715}
\end{align}
Here we used \eqref{eq:consistent2b} (applied with $j=1$, $i=2$, and $X=C_1/2$) to deduce that
$$\overline S_{2}(r)\E^{\Q_2}_r\left[\frac{C_2}{2}  \1_{\{S_{1,2}(T)<\infty\}}\right]=\overline S_{1}(r)\E^{\Q_1}_r\left[\frac{C_1}{2}  \1_{\{S_{1,2}(T)>0\}}\right].$$

Expressing \eqref{eq:220715} in terms of the first num\'eraire, we then get, on the event $\{1 \in \mathfrak A(r)\}$,
\begin{align} \label{eq:170310.2}
\frac{\E^{\overline \Q}[\overline C]}{\overline S_1(r)} = \E^{\Q_1}_r[C_1  ]+ S_{1,2}(r)\E^{\Q_2}_r[C_2  \1_{\{S_{1,2}(T)=\infty\}}].  
\end{align}
This corresponds exactly to the valuation formula in \citet{CFR2011}, constructed to restore put-call parity in a strict local martingale model.  The right-hand side of \eqref{eq:170310.2} is the sum of two terms.  The first term  is the valuation expectation of the contingent claim if the first currency is chosen as num\'eraire. The second term can be interpreted as a correction factor.  It is a product of the exchange rate, converting units of the second currency into units of the first currency, and another valuation expectation.  This time, the valuation expectation is taken using the second currency as num\'eraire.  It considers the contingent claim on the event where the first currency completely devalues. In the case when the contingent claim $C$ is a European call (with the first currency chosen as num\'eraire), this second term corresponds exactly to the ad-hoc correction in \citet{Lewis}.  Thus,  \eqref{eq:170310.2}
 retrieves exactly the valuation formulas in   \citet{Lewis}, \citet{MadanYor_Ito},  \citet{Paulot}, and \cite{Kardaras:2015:valuation}. We also refer to Section~6 in \cite{Herdegen:Schweizer:2017} for an alternative approach based on well chosen no-arbitrage principles.
 \qed
\end{ex}

We next study the extension of the Black-Scholes-Merton model proposed in \cite{Camara:Heston}. They suggest to augment the original Black-Scholes-Merton model by allowing  the relative prices to jump to zero and infinity.  The jump to zero ``adjust[s] the Black-Scholes model for biases related with out-of-the-money put options'', and the jump to infinity 
``captures the exuberance and the extreme upside potential of the market and leads to a risk-neutral density with more positive skewness and kurtosis than the density implicit in the Black-Scholes model.''
\cite{Camara:Heston} then illustrate that such a modification indeed yields an implied volatility which is closer to the ones observed in the market.
\begin{ex}[Black-Scholes with jumps to zero and infinity]\label{ex:camaraHeston}
We consider again two currencies, that is, $d=2$.  We assume that the relative prices are described through the Black-Scholes model; however, now with the additional feature that the price may either jump to zero or infinity at some exponential time.  We introduce the model formally by specifying a probability measure $\Prob$ on $(\Omega, \mathcal F(T))$. Towards this end, suppose that  $\tau_1$ and $\tau_2$ are exponentially distributed stopping times with intensity $\lambda^{\Prob}_1$ and $\lambda^{\Prob}_2$, respectively, and satisfy  $\Prob(\tau_1=\tau_2)=0$. 
We then set \[ S_{1,2}(t)= S_{1,2}(0)\exp\left(\sigma W(t) - \frac{\sigma^2}{2} t +\mu t \right)\1_{\{t< \tau_1 \wedge \tau_2\}}  + \infty \1_{\{\tau_1\leq t\wedge \tau_2\}} , \]
where $\mu$, $\sigma \in \R$ are constant with $\sigma \neq 0$ and $W$ is a $\Prob$--Brownian motion, independent of $\tau_1$ and $\tau_2$.
This yields directly
\[S_{2,1}(t)=S_{2,1}(0) \exp\left(-\sigma W(t) + \frac{\sigma^2}{2}t  -\mu t\right)\1_{\{t<\tau_1 \wedge \tau_2 \} }  + \infty \1_{\{\tau_2\leq t\wedge\tau_1 \}}.\]
Thus, on the event $\{\tau_1<\tau_2\}$, the first currency devalues completely at time $\tau_1$, while on $\{\tau_2<\tau_1\}$ the second currency devalues completely at time $\tau_2$.

We now want to construct a valuation measure  $\overline \Q$ with respect to the basket. Towards this end, we first construct a num\'eraire-consistent family of probability measures $(\Q_1, \Q_2)$ and then apply Theorem~\ref{T:extending_val_operator_prelim}\ref{T:extending_val_operator_prelim:b}.  
In particular, under $\Q_1$ the process $S_{1,2}$ stays real-valued and is a supermartingale; a similar statement holds for $\Q_2$. 
To start, we  define the probability measures $\Prob_1$ and $\Prob_2$ by 
\begin{align}
\frac{\dd  \Prob_1}{\dd \Prob}&=\frac{\1_{\{\tau_1>\tau_2 \wedge T\}}}{\Prob(\tau_1 > \tau_2 \wedge T|\tau_2)} = \1_{\{\tau_1>T\wedge\tau_2\}}\mathrm e^{\lambda^{\Prob}_1(T\wedge\tau_2)};  \label{eq:ProbExCamHeston1}\\
\frac{\dd  \Prob_2}{\dd \Prob} &=\frac{\1_{\{\tau_2>\tau_1 \wedge T\}}}{\Prob(\tau_2 > \tau_1 \wedge T|\tau_1)} = \1_{\{\tau_2>T\wedge\tau_1\}}\mathrm e^{\lambda^{\Prob}_2(T\wedge\tau_1)}.  \label{eq:ProbExCamHeston2}
\end{align}

We next fix some, for the moment arbitrary,  constants  $\mu_1, \mu_2 \in \R$ and $\lambda_1, \lambda_2>0$ and define the probability measures $\Q_1$ and $\Q_2$ by 
\begin{align}\label{eq:QprobExCamHeston1}
\frac{\dd \Q_1}{\dd  \Prob_1}&=\mathcal{E}\left(\left(\frac{\mu_1 - \mu}{\sigma}\right) W\right)(T)\,\, \mathrm e^{(\lambda^{\Prob}_2-\lambda_2)(T\wedge\tau_2)}\left(\frac{\lambda_2}{\lambda^{\Prob}_2}\right)^{\1_{\{\tau_2\leq T\}}};\\ 
\label{eq:QprobExCamHeston2}
\frac{\dd \Q_2}{\dd  \Prob_2}&=\mathcal{E}\left(\left(\frac{\mu_2 - \mu +\sigma^2}{\sigma}\right) W\right)(T)\,\,\mathrm e^{(\lambda^{\Prob}_1-\lambda_1)(T\wedge\tau_1)}\left(\frac{\lambda_1}{\lambda^{\Prob}_1}\right)^{\1_{\{\tau_1\leq T\}}}.
\end{align}
Then the $\Q_1$--intensity of $\tau_2$ equals $\lambda_2$ and the $\Q_2$--intensity of $\tau_1$ equals $\lambda_1$. Moreover, we get
\begin{align}\label{eq:SExCamHeston1}
S_{1,2}(t)&= S_{1,2}(0) \exp\left(\sigma W_1(t) -\frac{\sigma^2 }{2} t + \lambda_2 t \right)  \1_{\{t< \tau_2 \} } \mathrm e^{ \mu_1  t - \lambda_2 t }, \qquad \text{$\Q_1$--almost surely}; \\
\label{eq:SExCamHeston2}
S_{2,1}(t)&= S_{2,1}(0) \exp\left(\sigma W_2(t) -\frac{\sigma^2 }{2}  t + \lambda_1 t \right)  \1_{\{t< \tau_1 \} } \mathrm e^{-\mu_2 t -  \lambda_1  t }, \qquad \text{$\Q_2$--almost surely}
\end{align}
for all $t$,
with $W_1$ a $\Q_1$--Brownian Motion independent of $\tau_2$ and $W_2$ a $\Q_2$--Brownian motion independent of $\tau_1$.   
It is clear that it is necessary to have
$\lambda_1 \geq -\mu_2$ and $\lambda_2 \geq \mu_1$
for the supermartingale property of $S_{1,2}$ and $S_{2,1}$, respectively.

Fix now $t \in [0,T]$ and  $A\in\mathcal{F}(t)$. Then, by \eqref{eq:QprobExCamHeston1}--\eqref{eq:QprobExCamHeston2}, \eqref{eq:ProbExCamHeston1}--\eqref{eq:ProbExCamHeston2}, and~\eqref{eq:SExCamHeston1}--\eqref{eq:SExCamHeston2}
\begin{align*}
\Q_1(A\cap \{S_{1,2}(t)>0\}) &=\E^{\Prob}\left[\mathcal{E}\left(\left(\frac{\mu_1 - \mu}{\sigma}\right) W\right)(t)\,\,  \mathrm e^{(\lambda_1^{\Prob}+\lambda^{\Prob}_2 - \lambda_2 )t}\1_{\{t< \tau_1\wedge\tau_2 \} }\1_A\right];\\
S_{1,2}(0) \E^{\Q_2}[S_{2,1}(t)\1_A] &=\E^{\Prob}\left[\mathcal{E}\left(\left(\frac{\mu_2 - \mu}{\sigma}\right) W\right)(t)\,\,  \mathrm e^{(\lambda_1^{\Prob}+\lambda^{\Prob}_2 - \mu_2  - \lambda_1 ) t}\1_{\{t< \tau_1\wedge\tau_2 \} }\1_A\right].
\end{align*} 
This yields that for \eqref{eq:consistent2a} to hold we need to impose that
\[\lambda_2-\lambda_1=\mu_1 = \mu_2.\]
Indeed, this is  sufficient for the num\'eraire-consistency of $(\Q_1, \Q_2)$ because then 
also, in the same manner, 
\begin{align*}
S_{2,1}(0) \E^{\Q_1}[S_{1,2}(t)\1_A] &=
\Q_2(A\cap \{S_{2,1}(t)>0\}).
\end{align*} 
Theorem~\ref{T:extending_val_operator_prelim}\ref{T:extending_val_operator_prelim:b} now yields a valuation measure  $\overline \Q$ with respect to the basket, corresponding to the family  $(\Q_1, \Q_2)$.

Consider next an exchange option $C = (C_1, C_2)$ with $C_1=(S_{1,2}(T)-K)^+$ and  $C_2=(1-KS_{2,1}(T))^+$, where $K \in \R$. That is, at time $T$, the option gives the right to swap  $K$ units of the first currency into one unit of the second currency.  Then the representation of $\E^{\overline \Q}$ in \eqref{eq:220715} of Example~\ref{ex:cosistent_2assets} yields
\begin{align}
\E^{\overline \Q}[\overline C]&= \overline S_1(0) \E^{\Q_1}[(S_{1,2}(T)-K)^+\1_{\{\tau_2>T\}}]+ \overline S_{2}(0)\Q_2(\tau_1\leq T) \nonumber \\
&=  \overline S_1(0) \Q_1(\tau_2>T)\E^{\Q_1}\left[\left(S_{1,2}(0)\mathrm e^{\sigma W_1(T)+(\lambda_2-\lambda_1-{\sigma^2}/2)T}-K\right)^+\right] + \overline S_2(0)(1-{\rm e}^{-\lambda_1 T}) \nonumber\\
&= \overline S_2(0)  {\rm e}^{-\lambda_1 T} \Phi(d_1)- \overline S_1(0)  K{\rm e}^{-\lambda_2T}\Phi(d_2)+\overline S_2 (0)(1-{\rm e}^{-\lambda_1 T}),
\label{eq:camraheston1}
\end{align}
where 
\begin{equation*}
d_1=\frac{1}{\sigma\sqrt{T}}\left(\ln\left(\frac{S_{1,2}(0)}{K}\right)+\left(\lambda_2-\lambda_1+\frac{\sigma^2}{2}\right)T\right); \qquad
d_2=d_1-\sigma\sqrt{T}
\end{equation*}
and 
$\Phi$ is the standard normal cumulative distribution function. For the last equality in \eqref{eq:camraheston1}, we have used the standard Black-Scholes-Merton formula with interest rate $\lambda_2 - \lambda_1$. 

The expression in \eqref{eq:camraheston1}  corresponds to formula~(16) in \cite{Camara:Heston}.   That formula has been derived via solving a partial integral differential equation.  In contrast,  \eqref{eq:camraheston1}  has been derived by a purely probabilistic approach based on equivalent supermartingale measures.  Note that the use of a num\'eraire-consistent family yields a systematic way to value more complicated, possibly path-dependent contingent claims in the C\^amara-Heston framework.  Moreover, this example also shows that the C\^amara-Heston framework is free of arbitrage, in the sense that there exists a valuation measure $\overline \Q$ with respect to the basket.  Due to the presence of a jump to zero and due to the incompleteness of the model this example is not covered by \cite{CFR2011}.

We emphasize that this approach is not restricted to the Black-Scholes model. One might take any model, for example the Heston model, and then add a jump to zero and a jump to infinity. Going through the same steps as in this example then yields a num\'eraire-consistent family  that corrects deep out-of-the money puts and call prices. 
\qed
\end{ex}

\subsection{Aggegration without num\'eraire-consistency}  \label{SS:4.3}

Theorem~\ref{T:extending_val_operator_prelim}\ref{T:extending_val_operator_prelim:b} yields that,  given a num\'eraire-consistent family of probability measures $(\Q_i)_i$, there exists a  valuation measure $\overline \Q$ with respect to the basket. In practice it might be difficult to decide whether a  given family of probability measures $(\Q_i)_i$ is num\'eraire-consistent. Thus, the question arises, under which conditions the existence of a not necessarily num\'eraire-consistent family of probability measures yields the existence of a valuation measure $\overline \Q$ with respect to the basket.  The next theorem  provides more easily verifiable conditions  for an arbitrary family of probability measures  $(\Q_i)_i$.  It requires the following definition.

\begin{defin}[NOD]  \label{D:NOD}
We say that a probability measure $\Prob$ on $(\Omega, \mathcal F(T))$ satisfies No Obvious Devaluations (NOD) if $\Prob(i \in \mathfrak{A}(T) |\mathcal F(\tau)) >0$ on $\{\tau < \infty\} \cap \{i\in \mathfrak{A}(\tau)\}$, $\Prob$--almost surely,  for all $i$ and  stopping times $\tau$. \qed
\end{defin}

A probability measure $\Prob$ that satisfies NOD guarantees the following.  If at any point of time $\tau$ a certain currency $i$ has not yet defaulted then the probability is strictly positive that this currency will not default in the future. \citet{CFR2011} study the case $d=2$ and also introduce the notion of ``no obvious hyperinflations'', seemingly different.  That paper, however, has an additional standing hypothesis, namely that there are no sudden complete devaluations through a jump (see Definition~\ref{D:NSD} below). Under this condition, their notion of ``no obvious hyperinflations'' and this paper's notion NOD agree.

\begin{thm}[Aggregation without num\'eraire-consistency]\label{T:1'}
Let $(\Q_i)_i$ be a family of probability measures. Then there exists a  valuation measure $\overline \Q \sim \sum_i \Q_i$  with respect to the basket
if one of the following two conditions is satisfied.
	\begin{enumerate}[label={\rm(\alph{*})}, ref={\rm(\alph{*})}]
		\item\label{T:1':a}   $S_i$ is a $\Q_i$--martingale for each $i$.
		\item\label{T:1':b}   The following four conditions hold.
			\begin{enumerate}[label={\rm(\roman{*})}, ref={\rm(\roman{*})}]
					\item\label{T:1':0i}  $S_i$ is a $\Q_i$--local martingale for each $i$.
					\item\label{T:1':i}  $\sum_i \Q_i/d$ satisfies NOD; see Definition~\ref{D:NOD}.
					\item\label{T:1':ii}  For each $i$,
					$$\Q_i|_{\mathcal F \cap \{\sum_j S_{i,j}(T) < \infty\}} \sim \left.\sum_k \Q_k \right|_{\mathcal{F} \cap \{\sum_j S_{i,j}(T) < \infty\}}.$$
					\item\label{T:1':iii}  There exist $\epsilon>0$, $N\in \N$, and predictable times $(T_n)_{n \in \{1, \cdots, N\}}$ such that
\begin{align*}
	\bigcup_i \left\{(t,\omega): \sum_j S_{i,j}(t)=\infty \text{  and  }  \sum_j S_{i,j}(t-) \leq d+\varepsilon \right\}  \subset \bigcup_{n =1}^N \lc T_n\rc,
\end{align*}
$(\sum_i \Q_i/d)$--almost surely.
	\end{enumerate}
\end{enumerate}
\end{thm}

As Example~\ref{Ex:4.9} below illustrates, Theorem~\ref{T:1'}\ref{T:1':b} is not sufficient for the existence of a valuation measure for the basket, in general, without \ref{T:1':b}\ref{T:1':0i}, namely that  $S_i$ is a $\Q_i$--local martingale for each $i$.  The condition in Theorem~\ref{T:1'}\ref{T:1':b}\ref{T:1':i} states that $\sum_i \Q_i/d$ must satisfy the minimal no-arbitrage condition given by NOD --- the selling of an active currency does not yield a simple arbitrage strategy.  Indeed, it is clear that this condition  is necessary.  As Example~\ref{Ex:2aa} below illustrates, the conclusion of   Theorem~\ref{T:1'} is wrong without \ref{T:1':b}\ref{T:1':i}. Thus, given the other conditions, it is not redundant for the formulation of the theorem. The condition in Theorem~\ref{T:1'}\ref{T:1':b}\ref{T:1':ii}  means that the support of $\Q_i$ is the event $\{\sum_{j} S_{i,j}(T)<\infty\}$ for each $i$. The necessity of such a condition is the content of Example~\ref{Ex:2} below. 

Theorem~\ref{T:1'}\ref{T:1':b}\ref{T:1':iii} is a technical condition and we do not know whether it is necessary for the statement of the theorem to hold.    This condition allows the $k$-th currency to devalue suddenly, as long as it is not  ``strong''  in the sense  $\sum_j S_{i,j} \leq d+\varepsilon$.   If, however, a ``strong'' currency devalues suddenly, it only is allowed to do so at a finite number of fixed, predictable times.  In particular, any discrete-time model with finitely many time steps satisfies this condition.  This condition also holds if $\sum_i \Q_i/d$ satisfies NSD, in the sense of the following definition. 

\begin{defin}[NSD]\label{D:NSD}
We say that a probability measure $\Prob$ satisfies No Sudden Devaluation  (NSD) if $\Prob(S_{i,j} \text { jumps to }\infty) = 0$ for all $i,j$.
\qed
\end{defin}

Under NSD no currency devalues completely against any other currency suddenly.
 Example~\ref{Ex:2aa} below illustrates that there exists a probability measure $\Prob$ that satisfies NSD but not NOD. 	 It is simple to construct an example that satisfies NOD but not NSD.

\begin{ex}[On the necessity of  Theorem~\ref{T:1'}\ref{T:1':b}\ref{T:1':0i}]  \label{Ex:4.9}
Fix $T=d=2$ and $\Omega = \{\omega_1, \omega_2\}$ along with $\mathcal F(t) = \{\emptyset, \Omega\}$ for all $t <1$ and $\mathcal F(t) = \{\emptyset, \Omega, \{\omega_1\}, \{\omega_2\}\}$ for all $t \geq 1$. Let $S_{1,2}(\omega_1, t) = 1$ and $S_{1,2}(\omega_2, t) \equiv \1_{t <1}$ for all $t$.  That is, two states of the world are possible; up to time $1$ the exchange rate between the two currencies stays constant, and at time one either the second currency devalues completely or nothing happens, depending on the state of the world. We now let  $\Q_1(\{\omega_1\}) = \Q_1(\{\omega_2\} )= 1/2$, and $\Q_2(\{\omega_1\} ) = 1$.  Then $S_{1,2}$ is a strict $\Q_1$--supermartingale and $S_{2,1}$ is a $\Q_2$--martingale. Moreover, all conditions in Theorem~\ref{T:1'}\ref{T:1':b}, apart from  \ref{T:1':0i}, are satisfied. However, selling one unit of the second currency and buying one unit of the first currency at time zero yields a nonnegative wealth process that is strictly positive in state $\omega_2$, which has strictly positive $(\Q_1+\Q_2)/2$--probability; thus a clear arbitrage. Thus, no valuation measure $\overline \Q \sim (\Q_1 + \Q_2)$  with respect to the basket can exist.  This illustrates that Theorem~\ref{T:1'} indeed needs the local martingale property, formulated in \ref{T:1':b}\ref{T:1':0i}, in its statement. \qed
\end{ex}

\begin{ex}[On the necessity of  Theorem~\ref{T:1'}\ref{T:1':b}\ref{T:1':i}]  \label{Ex:2aa}
 We slightly modify Example~\ref{Ex:4.9}. Again, fix  $T=d=2$ and assume that $(\Omega, \mathcal F(T), \Q_1)$ supports a Brownian motion $B$ started in zero and stopped when hitting $-1$, and an independent $\{0, 1\}$--distributed random variable $X$ with $\Q_1(X=0) = \Q_1(X=1) = 1/2$. Now, let $$S_{1,2}(t)  = 1 + \1_{\{X=1\}} B\left( \tan\left(0\vee \frac{\pi}{2}(t-1)\right)\right)$$ 
for all $t$. Now let $(\mathcal F(t))_t$ denote the smallest right-continuous filtration that makes $S_{1,2}$ adapted. Then $S_{1,2}$ is constant before time one and stays constant afterwards with probability $1/2$, but moves like a time-changed Brownian motion stopped when hitting zero, otherwise. We now set  $\Q_2 = \Q_1(\cdot | \{X=0\})$ and note that $S_{2,1}$ is a (constant) $\Q_2$--martingale. Then the conditions in  Theorem~\ref{T:1'}\ref{T:1':b}\ref{T:1':0i}, \ref{T:1':ii}, and \ref{T:1':iii} are all satisfied, but the same strategy as in the previous example yields an arbitrage, admissible when the basket is chosen to be the num\'eraire.  Thus, Theorem~\ref{T:1'}\ref{T:1':b}\ref{T:1':i} is necessary to make the theorem valid.  Note that  $(\Q_1 + \Q_2)/2$ satisfies NSD  but not NOD in this example. 
\qed
\end{ex}

\begin{ex}[On the necessity of  Theorem~\ref{T:1'}\ref{T:1':b}\ref{T:1':ii}]  \label{Ex:2}
With $d=2$ assets again, we now provide an example for a family  of  local martingale measures $(\Q_1, \Q_2)$ such that  $(\Q_1 + \Q_2)/2$ satisfies NSD and NOD, but no valuation measure $\overline \Q \sim (\Q_1 + \Q_2)$ exists with respect to the basket.
Fix $T=2$ and a filtered probability space $(\Omega, \mathcal F(2), (\mathcal F(t))_t, \Q_2)$ that supports a three-dimensional Bessel process $R$ starting in one. Next, let $\tau$ denote the smallest time that $R$ hits $1/2$; in particular, we have $\Q_2(\tau <T)>0$ and $\Q_2(\tau = \infty) >0$ .   Consider now the process 
\begin{align*}
	S_{1,2}= 1 + \left(R-\frac{1}{2}\right) \1_{\lc \tau, \infty\lc} >0.
\end{align*}
With  $\Q_1(\cdot) = \Q_2(\cdot|\{\tau = \infty\})$ we have $\Q_1(S_{1,2} = 1)=1$. Moreover,  $S_{2,1}$ is a $\Q_2$--local martingale and  $\mathfrak{A}(T) = \{1,2\}$. In particular, $(\Q_1 + \Q_2)/2$ satisfies NSD and NOD. 
However,  Proposition~\ref{P:cons}\ref{P:cons:3} yields that no num\'eraire-consistent family of probability measures can exist. Thus, 
 Theorem~\ref{T:extending_val_operator_prelim}\ref{T:extending_val_operator_prelim:a} yields that no valuation measure $\overline \Q \sim  (\Q_1 + \Q_2)$  with respect to the basket exists either. This shows that Theorem~\ref{T:1'}\ref{T:1':b} is not correct without the support condition in \ref{T:1':b}\ref{T:1':ii}.
\qed
\end{ex}

 \section{Proofs  of Proposition~\ref{P:cons} and Theorems~\ref{T:extending_val_operator_prelim}  and \ref{T:1'}}   \label{S:proofs}

\begin{proof}[Proof of Proposition~\ref{P:cons}] In the following we argue the three parts of the statement. 

	 \ref{P:cons:1}:  Fix $i$ and $j$ and  note that \eqref{eq:consistent2a} yields that $\Q_i(i\notin{\mathfrak{A}(t)})=0$ for all $t$.    Monotone convergence then yields
\begin{equation}\label{eq:consistent2a'}
        \E^{\Q_i}[S_{i,j}(t) X]= S_{i,j}(0) \E^{\Q_j}[ X \1_{\{S_{j,i}(t)>0\}}]
    \end{equation}
    for all bounded $\mathcal F(t)$--measurable random variables  $X$ and for all $ t$.   To show  \eqref{eq:consistent2b}, fix a bounded $\mathcal F(t)$--measurable random variable $X$, $A \in \mathcal F(r)$, and 
	$r \leq t$. We then have
    \begin{align*}
        \E^{\Q_i}[S_{i,j}(t) X  \1_{A} ] &=   \E^{\Q_i}[S_{i,j}(t) X  \1_{A} \1_{{\{S_{j,i}(r)>0\}}}  ]  = S_{i,j}(0) \E^{\Q_j}[X \1_{\{S_{j,i}(t)>0\}} \1_A \1_{{\{S_{j,i}(r)>0\}}} ]  \\
        		&=  S_{i,j}(0) \E^{\Q_j}[\E_r^{\Q_j} [X \1_{ {\{S_{j,i}(t)>0\}}}]  \1_A \1_{{\{S_{j,i}(r)>0\}}}  ] \\
		&= \E^{\Q_i}[S_{i,j}(r) \E_r^{\Q_j} [X \1_{ {\{S_{j,i}(t)>0\}}}]   \1_{A} ]
    \end{align*}	
    by applying \eqref{eq:consistent2a'} twice, which yields  \eqref{eq:consistent2b}. The fact that $S_i$ is a $\Q_i$--supermartingale  follows from \eqref{eq:consistent2b} with $X=1$.
    
     \ref{P:cons:2}:  Fix $i$ and $j$. As in Proposition~2.3 in \citet{Perkowski_Ruf_2014}, we may replace $t$ in \eqref{eq:consistent2a}  by a stopping time $\tau$.  With $A = \Omega$, we then have
     \begin{align*}
             \E^{\Q_i}[S_{i,j}(\tau) ]= S_{i,j}(0) {\Q_j}(S_{j,i}(\tau)>0)
     \end{align*}
     for all stopping times $\tau$.
     Recall now that $S_{i,j}$ is  a $\Q^i$--supermartingale and localize with a sequence of first crossing times.
     
         \ref{P:cons:3}: The first part follows as in \ref{P:cons:2}. The second statement follows directly from \eqref{eq:consistent2b}.
\end{proof}

\begin{proof}[Proof of Theorem~\ref{T:extending_val_operator_prelim}]
{ ~}

\ref{T:extending_val_operator_prelim:a} \& \ref{T:extending_val_operator_prelim:c}: Let  $\overline \Q$ be a valuation measure with respect to the basket.  Define now a family $(\Q_i)_i$ of  probability measures $(\Q_i)_i$  by
$
{\dd\Q_i}/{\dd \overline \Q}={\overline{S}_i(T)}/{\overline{S}_i(0)}.
$
Thanks to \eqref{eq:170310.3}, we have $\sum_i \overline{S}_i(0)  {\Q_i}=\overline \Q$ and thus $\overline \Q \sim \sum_i \Q_i$. Moreover,  thanks to \eqref{eq:160727.1} we have 
 \begin{align*}
	\E^{\Q_i}[S_{i,j}(t) \1_A] &=  (\overline{S}_{i}(0))^{-1} \E^{\overline \Q} [S_{i,j}(t) \1_A  \overline{S}_{i}(t)  \1_{\{S_{j,i}(t)>0\}}]   \\
		&= (\overline{S}_{i}(0))^{-1}  \E^{\overline \Q}[\overline{S}_{j}(t) \1_A \1_{\{S_{j,i}(t)>0\}}] \\
        &= (\overline{S}_{i}(0))^{-1}  \overline{S}_{j}(0)   {\Q_j}(A \cap \{S_{j,i}(t)>0\}) = S_{i,j}(0)   {\Q_j}(A \cap \{S_{j,i}(t)>0\})
    \end{align*}
for all $i,j$,  $A\in \mathcal F(t)$, and $t$. Hence, the family $(\Q_i)_{i}$   is num\'eraire-consistent.
 
Next, suppose that $\overline C\in L^1(\overline \Q)$. The representation of $\overline C$ given by~\eqref{def:basketvalue} and Bayes' formula yield
\begin{align*}
		\E^{\overline \Q}_r[\overline C ] & = \sum_j \E^{\overline \Q}_r\left[ \frac{C_j}{|\mathfrak{A}(T)|} \overline S_j(T)  \1_{\{j \in \mathfrak A(T)\}} \right]  =\sum_{j} \E^{\Q_{j}}_r\left[\frac{C_j}{|\mathfrak{A}(T)|}\right]  \overline S_j(r)   \1_{\{j \in \mathfrak A(r)\}}; 
	\end{align*}
hence \eqref{eq:extension_val_operator} follows. If $\overline C = \overline C \1_{\{i \in \mathfrak A(T)\}}$, $\overline \Q$--almost surely,  for some $i$, then the same computation, in conjunction with $C_j \overline S_j(T) = C_i \overline S_{i}(T)$ which follows from \eqref{eq:160727.1} for all $j \in \mathfrak A(T)$,
yields
 \eqref{eq:220815}.  The uniqueness of $(\Q_i)_i$  follows from \eqref{eq:220815} with $r = 0$ and $\overline C = \1_{A \cap \{i \in \mathfrak A(T)\}}$ for all $A \in \mathcal F(T)$. 

\ref{T:extending_val_operator_prelim:b}: The uniqueness is clear. For existence, define $\overline \Q = \sum_i \overline{S}_i(0)  {\Q_i}$.  It is clear again that $\overline \Q\sim  \sum_i \Q_i$. By the computations in the first part of the proof, we only need to argue that the process $\overline{S}$ is a $\overline \Q$--martingale and
$
{\dd\Q_i}/{\dd \overline \Q}={\overline{S}_i(T)}/{\overline{S}_i(0)}.
$ 
To this end, it suffices to  show that
\begin{equation}\label{L:consistent extended}
\E^{\overline \Q}[\overline{S}_{i}(t)\1_A]=\overline{S}_{i}(0) { \Q_i}(A)
\end{equation}
for all $i,t$ and $A\in\mathcal{F}(t)$. Indeed \eqref{eq:160727.1} and Proposition~\ref{P:cons}\ref{P:cons:1} imply
 \begin{align*}
	                \E^{\overline \Q}[\overline{S}_{i}(t)\1_A]&= \sum_j \overline{S}_{j}(0) \E^{\Q_j}[\overline{S}_{i}(t)\1_A ]
	                = \sum_j \overline{S}_{i}(0)S_{i,j}(0) \E^{\Q_j}[\overline{S}_{i}(t)\1_A\1_{\{S_{j,i}(t)>0\}}]\\
	       &= \sum_j \overline{S}_{i}(0) \E^{\Q_i}[S_{i,j}(t) \overline{S}_{i}(t)\1_A]
	       	       = \sum_j \overline{S}_{i}(0) \E^{\Q_i}[ \overline{S}_{j}(t)\1_A] = \overline{S}_{i}(0) {\Q_i}(A),     
    \end{align*}
yielding \eqref{L:consistent extended}.
\end{proof}

\begin{proof}[Proof of Theorem~\ref{T:1'}] 
{  ~ }

\ref{T:1':a}: Consider the probability measures $\widetilde{\Q}_i$ given by $\dd \widetilde{\Q}_i/\dd\Q_i = \sum_j  S_{i,j}(T) / \sum_j  S_{i,j}(0)$ for each $i$,  and $\overline \Q = \sum_i \widetilde{\Q}_i/ d$.  Then 
we have $\overline\Q  \sim \sum_j \Q_j$. Moreover, $\overline{S}$ is a $\widetilde \Q_i$--martingale for each $i$, thus it is also a $\overline \Q$--martingale. 

\ref{T:1':b}: We set $\Prob = \sum_i \Q_i/d$ and fix $\varepsilon>0$ as in \ref{T:1':b}\ref{T:1':iii}.  To prove the statement is suffices to construct a strictly positive $\Prob$--martingale $Z$ such that $Z  \overline{S}$ is also a  $\Prob$--martingale.  We proceed in several steps.

\emph{Step~1}:  For the construction of $Z$ below, we shall iteratively pick the strongest currency until some time when it is not the strongest anymore, at which point we switch to the new strongest one.  To follow this program, define the sequences of stopping times $(\tau_n)_{n \in \N_0}$ and currency identifiers $(i_n)_{n \in \N}$ by $\tau_0 = 0$ and
	\begin{align}
		i_n &= \arg \min_{i \in \{1, \ldots, d\}} \{ (\overline{S}_{i}(\tau_{n-1}))^{-1}\};  \label{E in}\\
		\tau_n &= \inf\{t \in [\tau_{n-1},T]: (\overline{S}_{i_n} (t))^{-1}> d+\varepsilon\} \label{eq:250815}
	\end{align}
	for all $n \in \N$, where possible conflicts  in \eqref{E in} are solved by lexicographic order.
	
\emph{Step~2}: We claim that  	$\Prob(\lim_{n \uparrow \infty} \tau_n > T) = 1$.
To see this, assume that  $\Prob(\lim_{n \uparrow \infty} \tau_n \leq T) > 0$. Then there exist $i$ and $j$ such that  $(\overline{S}_{i})^{-1}$ has infinitely many upcrossings from $d$ to $d + \varepsilon$  with strictly positive $\Q_j$--probability.  Next,  by a simple localization argument we may assume that $(\overline{S}_{j})^{-1}$ is a $\Q_j$--martingale and we consider the corresponding measure $\widehat\Q$, given by $\dd \widehat\Q/ \dd \Q_j = \overline{S}_{j}(0)(\overline{S}_{j}(T))^{-1}$.  Note that $\widehat{\Q} \sim \Q_j$ and that the process $\overline{S}_{i}$  is a bounded $\widehat{\Q}$--martingale that has infinitely many downcrossings from $1/d$ to $1/(d+\varepsilon)$ with positive probability. This, however, contradicts the supermartingale convergence theorem, which then yields a contradiction. Thus, the claim holds.

\emph{Step~3}: Assume that we are given a nonnegative stochastic process $Z$ such that   $Z^{\tau_n}$ and $Z^{\tau_n}  \overline{S}^{\tau_n}$  are $\Prob$--martingales for each $n\in \N$, in the notation of \eqref{eq:250815}.  We then claim that $Z$ and $Z  \overline{S}$ are $\Prob$--martingales.  To see this, note that $Z$ and $Z \overline{S}$  are $\Prob$--local martingales by \emph{Step~2}.  
Next, define a sequence of probability measures $(\Q^n)_{n \in \N}$ via $\dd \Q^n / \dd \Prob = Z^{\tau_n}(T)$ and note that $\overline{S}^{\tau_n}$ is a $\Q^n$--martingale satisfying $\overline{S}_{i_n}(\tau_{n-1}) \geq 1/d$ on the event $\{\tau_{n-1} \leq T\}$, where $i_n$ is given in \eqref{E in}.  Thus, on  $\{\tau_{n-1} \leq T\}$ we have
\begin{align*}
	\frac{1}{d} \leq \E^{\Q^n}\left[ \overline{S}_{i_n}(\tau_n)| \mathcal F(\tau_{n-1})\right] \leq 1-q_n + \frac{q_n}{d+\varepsilon},
\end{align*}
where $q_n = \Q^n(\tau_n \leq T|\mathcal F(\tau_{n-1}))$, for each $n \in \N$.  We obtain that
\begin{align*}
	q_n \leq \frac{d^2 + \varepsilon d-d-\varepsilon} {d^2 + \varepsilon d-d} = \eta \in (0,1),
\end{align*}
which again yields
\begin{align*}
	\Q^n(\tau_n \leq T) &\leq \E^{\Q^n}\left[\Q^n(\left.\tau_n \leq T\right|\mathcal F(\tau_{n-1})) \1_{\{\tau_{n-1} \leq T\}}\right]
		\leq \eta \Q^n\left(\tau_{n-1} \leq T\right) \leq \eta^n
\end{align*}
for each $n \in \N$,
where the last inequality follows by induction.  This yields $\lim_{n \uparrow \infty} \Q^n(\tau_n \leq T) = 0$.  Now, a simple extension of Lemma~III.3.3 in \citet{JacodS}, such as the one of Corollary~2.2 in \citet{Blanchet_Ruf_2012}, yields that $Z$ is a $\Prob$--martingale.  As $\overline{S}$ is bounded, also  $Z \overline{S}$ is a $\Prob$--martingale.

\emph{Step~4}: We now construct a stochastic process $\widetilde Z$ that satisfies the assumptions of \emph{Step~3}. Towards this end, for each $i$, let $Z_i$ denote the unique $\Prob$--martingale such that $\dd \Q_i / \dd \Prob = Z_i(T)$. 
With the notation of \eqref{eq:250815}, \ref{T:1':b}\ref{T:1':i} and  \ref{T:1':ii} yield that $Z_{i_n}(\tau_{n-1}) > 0$ for each $n \in \N$.   This allows us to define the process $\widetilde Z$ inductively by 
$\widetilde Z(0)=1$ and 
	\begin{align*}
		\widetilde Z(t) = \widetilde Z(\tau_{n-1}) \times \frac{(\overline{S}_{i_n}(t))^{-1} \1_{\{Z_{i_n}(t) > 0\}} Z_{i_n}(t)}{(\overline{S}_{i_n}(\tau_{n-1}))^{-1} Z_{i_n}(\tau_{n-1})}
	\end{align*}
	 for all $t \in (\tau_{n-1}, \tau_n \wedge T]$ and $n \in \N$. Here we have again used the indices $(i_n)_{n \in \N}$  of \eqref{E in}.
	 As 
	 \begin{align*}
	 	\E^{\Prob}[(\overline{S}_{i_n}(\tau_n))^{-1} \1_{\{Z_{i_n}(\tau_n) > 0\}} Z_{i_n}(\tau_n) | \mathcal F(\tau_{n-1})]
			&= 	 \E^{\Q_{i_n}}[(\overline{S}_{i_n}(\tau_n))^{-1}  | \mathcal F(\tau_{n-1})] Z_{i_n}(\tau_{n-1})\\
			&=  (\overline{S}_{i_n}(\tau_{n-1}))^{-1} Z_{i_n}(\tau_{n-1})
	\end{align*}
	on $\{\tau_{n-1} \leq T\}$, the process
	 $\widetilde Z^{\tau_n}$ is a $\Prob$--martingale for each $n \in \N$.  
	 We now fix $j$ and argue that 	 $\overline{S}_{j}^{\tau_n}\widetilde Z^{\tau_n}$ is a $\Prob$--martingale for each $n \in \N$.  Thanks to \eqref{eq:160727.1} we have 
	\begin{align*}
		\overline{S}_{j}(t) \widetilde Z(t) = \overline{S}_{j}(\tau_{n-1}) \widetilde Z(\tau_{n-1}) \times \frac{S_{i_n,j}(t) \1_{\{Z_{i_n}(t) > 0\}} Z_{i_n}(t)}{S_{i_n,j}(\tau_{n-1}) Z_{i_n}(\tau_{n-1})}
	\end{align*}	 
 for all $t \in (\tau_{n-1}, \tau_n \wedge T]$  on $\{\overline{S}_{j}(\tau_{n-1}) > 0\}$ and $n \in \N$.  As zero is an absorbing state for $\overline{S}_{j}$ under $\Prob = \sum_i \Q_i/d$ the same arguments as above yield that $\overline{S}_{j}^{\tau_n}\widetilde Z^{\tau_n}$ is a $\Prob$--martingale for each $n \in \N$.

\emph{Step~5}: If $\Prob$ satisfies NSD, then $\widetilde{Z}$ is strictly positive because $\widetilde Z_{i_n}(\tau_n)>0$ for each $n \in \N$, in the notation of \eqref{E in} and \eqref{eq:250815}.  In this case, the proof of \ref{T:1':b} is finished.  However, under the more general condition in  \ref{T:1':b}\ref{T:1':iii} it cannot be guaranteed that the $\Prob$--martingale $\widetilde{Z}$ is strictly positive as it might jump to zero on  $\bigcup_{n \in \N} \lc\tau_n\rc  \bigcap  \bigcup_{m \in \{1, \cdots, N\}} \lc T_m\rc$. To address this issue, we shall modify the construction in \emph{Step~4} at the predictable times  $(T_m)_{m \in \{1, \cdots, N\}}$
to obtain a strictly positive  $\Prob$--martingale $Z$ such that also $Z \overline{S}$ is a $\Prob$--martingale. 

\emph{Step~5A}: We may assume that $0  < T_m < T_{m+1}$  on $\{T_m < \infty\}$ for all $m \in \{1, \cdots, N\}$ and, set, for sake of notational convenience, $T_0 = 0$ and $T_{N+1} = \infty$.   
In \emph{Step~5B}, we shall construct a family of strictly positive $\Prob$--martingales $(Y_m)_{m \in \{1, \cdots, N+1\}}$ that satisfy the following two conditions:
	\begin{itemize}
		\item $Y_m = Y_m^{T_m}$ and $Y_m^{T_{m-1}} = 1$; and
		\item $Y_m \overline{S}^{T_m} -  \overline{S}^{T_{m-1}}$ is a $\Prob$--martingale for all $m \in \{1, \cdots, N+1\}$.
	\end{itemize}
If we have such a family then the process $Z = \prod_{m=1}^{N+1} Y_m$ is a strictly positive $\Prob$--martingale and $Z \overline{S}$  a nonnegative $\Prob$--martingale. This then concludes the proof. 

\emph{Step~5B}: 
In order to construct a family of strictly positive $\Prob$--martingales $(Y_m)_{m \in \{1, \cdots, N+1\}}$ as desired, let us fix some $m \in \{1, \cdots, N+1\}$.   We first define a process $\widetilde Y$ by $\widetilde Y_m = 1$ on $\lc 0, T_{m-1}\rc$ and then by proceeding exactly as in \emph{Step~4}, but with $\tau_0= 0$ replaced by $\tau_0 = T_{m-1}$, with $\overline{S}$ replaced by $\overline{S}^{T_m}$ and with $Z_i$ replaced by $Z^{T_m}_i$ for each $i$.  Then $\widetilde Y_m$ is a nonnegative $\Prob$--martingale that   satisfies the two conditions of \emph{Step~5A}.
Let  $\widetilde{M}$ now denote the stochastic logarithm of $\widetilde Y_m$ and  $M_i$ the   stochastic logarithm of $(\overline{S}_{i})^{-1} Z_i$ for each $i$. 
 Note that, for each $i$, $M_i$ is only defined up to the first time that $(\overline{S}_{i})^{-1} Z_i$ hits zero, see also \citet{Larsson:Ruf:2017}.  Next, define the stochastic process
$$M = \widetilde{M}+  \left(\frac{1}{|\mathfrak{A}(T_m-)|} \sum_{j \in \mathfrak{A}(T_m-)} \Delta M_j(T_m)  - \Delta \widetilde{M}(T_m)\right) \1_{\lc T_m, \infty\lc };
$$
that is, $M$ equals $\widetilde{M}$ apart from the modification at time $T_m$ on $\{T_m<\infty\}$, where we replace its jump by the average jumps of the deflators corresponding to the active currencies at this point of time.   Then we have $\Delta M>1$, which implies that its stochastic exponential $Y_m = \mathcal E(M)$ is strictly positive. 
As the stopping time $T_m$ is predictable, the predictable stopping theorem implies that  $Y_m$ is a   $\Prob$--martingale and the two conditions in \emph{Step~5A} are satisfied.
\end{proof}

\bibliographystyle{apalike}
\setlength{\bibsep}{1pt}
\small\bibliography{aa_bib}{}
\end{document}